\documentclass[letterpaper,twocolumn,twoside]{IEEEtran} 

\usepackage{amsmath, amssymb,bm}
\usepackage{enumerate,mathtools}
\usepackage{amsthm}
\usepackage{cite}
\usepackage{microtype}

\theoremstyle{definition}

\newtheorem{theorem}{Theorem}
\newtheorem{cor}[theorem]{Corollary}
\newtheorem{lemma}[theorem]{Lemma}
\theoremstyle{definition}
\newtheorem{definition}{Definition}
\newtheorem{remark}{Remark}
\newtheorem{assumption}{Assumption}

\newtheorem{example}{Example}

\newcommand{\cC}{\mathcal{C}}

\newcommand{\cE}{\mathcal{E}}

\global\long\def\dd{\mathrm{d}}

\newcommand{\one}{\boldsymbol{1}}

\newcommand{\bEx}{\ensuremath{\mathbb{E}}}

\newcommand{\ex}[1]{\ensuremath{\mathbb{E}\left[ #1\right]}}
\newcommand{\pr}[1]{\ensuremath{\mathbb{P}\left[ #1\right]}}

\DeclareMathOperator{\var}{\sf Var}

\newcommand{\DKL}[2]{\ensuremath{D\left( #1 \, \|  \, #2 \right)}}

\newcommand{\Pmin}{\gamma_\mathrm{min}}
\newcommand{\Pmax}{\gamma_\mathrm{max}}

\newcommand{\Vg}{V_g} 
\newcommand{\Va}{V_a} 

\newcommand{\set}{\cE}

\DeclareMathOperator{\gtr}{tr}

\newcommand{\reals}{\mathbb{R}}

\newcommand{\eps}{\epsilon}
\newcommand{\normal}{\mathcal{N}}
\newcommand{\gmid}{\! \mid \!}

\let\originalleft\left
\let\originalright\right
\renewcommand{\left}{\mathopen{}\mathclose\bgroup\originalleft}
\renewcommand{\right}{\aftergroup\egroup\originalright}

\title{Conditional Central Limit Theorems\\ for Gaussian Projections}
\author{Galen Reeves
\thanks{The work of G.\ Reeves was supported in part by funding from the Laboratory for Analytic Sciences (LAS).  Any opinions, findings, conclusions, and recommendations expressed in this material are those of the author and do not necessarily reflect the views of the sponsors. }
\thanks{G.~Reeves is with the Department of Electrical and Computer Engineering and the Department of Statistical Science, Duke University, Durham, NC 27708 USA (e-mail: galen.reeves@duke.edu).}
}
 
\begin{document}

\maketitle

\begin{abstract}
This paper addresses the question of when projections of a high-dimensional random vector are approximately Gaussian. This problem has been studied previously in the context of  high-dimensional data analysis, where the focus is on low-dimensional projections of high-dimensional point clouds. The focus of this paper is on the typical behavior when the projections are generated by an i.i.d.\ Gaussian projection matrix. The main results are bounds on the deviation between the conditional distribution of the projections and a Gaussian approximation, where the conditioning is on the projection matrix. The bounds are given in terms of the quadratic Wasserstein distance and relative entropy  and are stated explicitly as a function of the number of projections and certain key properties of the random vector. The proof uses Talagrand's transportation inequality and a general integral-moment inequality for mutual information. Applications to random linear estimation and compressed sensing are discussed.  
\end{abstract}


\section{Introduction}

A somewhat surprising phenomenon is that the distributions of certain weighted sums  (or projections) of random variables can be close to Gaussian, even if the variables themselves have a nontrivial dependence structure. This fact can be traced back to the work of  Sudakov~\cite{sudakov:1978}, who showed that under mild conditions, the distributions of most one-dimensional projections of a high-dimensional vector are close to Gaussian. An independent line of work by Diaconis and Freedman \cite{diaconis:1984} provides similar results for projections of high-dimensional point clouds. In both cases, it is shown that the phenomenon persists with high probability when the weights are drawn randomly from the uniform measure on the sphere. Ensuing work \cite{hall:1993, weizsacker:1997, anttila:2003,bobkov:2003,naor:2003,dasgupta:2006, klartag:2007, klartag:2007a, meckes:2010aa,meckes:2012, dumbgen:2013, leeb:2013} has  generalized and strengthened these results in several directions, including the case of multivariate projections. 

Most related to the current paper is the recent line of work by  Meckes \cite{meckes:2010aa, meckes:2012}, who provides bounds with respect to the bounded-Lipschitz metric when the projections are distributed uniformly on the Stiefel manifold. Meckes shows that, under certain assumptions on a sequence of $n$-dimensional random vectors,  the distribution of the projections  are close to  Gaussian  provided that the number of projections $k$ satisfies  $k < 2 \log n / \log \log n$. Meckes also shows that this condition cannot be improved in general.

The focus of this paper is on the typical behavior when the projections are generated randomly and independently of the random variables.  Given an $n$-dimensional random vector $X$,  the $k$-dimensional linear projection $Z$ is defined according to
\begin{align}
Z = \Theta X, \label{eq:Z_def}
\end{align}
where $\Theta$ is a $k \times n$ random matrix that is independent of $X$. Throughout this paper it  assumed that $X$ has finite second moment and that the entries of $\Theta$ are i.i.d.\ Gaussian random variables with mean zero and variance $1/n$.

Our main results are bounds on the deviation between the conditional distribution of $Z$  given $\Theta$ and a Gaussian approximation. These bounds are given in terms of the quadratic Wasserstein distance and relative entropy and are stated explicitly as a function of the number of projections $k$ and certain properties of the distribution on $X$. For example, under the same assumptions used by Meckes \cite[Corollary~4]{meckes:2012}, we show that
\begin{align*}
\ex{ W_2^2(P_{Z \mid \Theta},G_Z)}   \le  C\, \left( n^{-\frac{1}{4}} +  k \,   n^{-\frac{2}{k+4}} \right),
\end{align*}
where the expectation is with respect to the random matrix $\Theta$, $W_2(\cdot, \cdot)$ denotes the quadratic Wasserstein distance, and $G_Z$ is the Gaussian distribution with the same mean and covariance as $Z$.

In comparison with previous work, one of the contributions of this paper is that our results provide a stronger characterization of the approximation error. Specifically, the analysis requires fewer assumptions about the distribution of $X$ and the bounds are stated with respect to stronger measures of statistical distance, namely the quadratic Wasserstein distance and relative entropy.  

A further contribution of the paper is given by our proof technique, which appears to be quite different from previous approaches. The first step in our proof is to characterize the conditional distribution of $Z$ after it has been passed through an additive white Gaussian noise (AWGN) channel of noise power $t \in (0,\infty)$. In particular,  the $k$-dimensional random vector $Y$ is defined according to
\begin{align}
Y =Z + \sqrt{t} N, \label{eq:Y_def} 
\end{align} 
where $N \sim \normal(0,I_k)$ is independent of $Z$. The bulk of the work is to bound the relative entropy between the conditional distribution of $Y$ given $\Theta$ and the Gaussian distribution with the same mean and covariance as $Y$. To this end, we take advantage of a general integral-moment inequality (Lemma~\ref{lem:MI_to_M_kp}) that allows us to bound the mutual information $I(Y; \Theta)$ in terms of the variance of the density function of $P_{Y | \Theta}$. Note that this density is guaranteed to exists because of the added Gaussian noise. 

The next step in our proof is to use the fact that the square of the conditional density can be expressed as an expectation with respect to two independent copies of $X$ using the identity:
\[
p^2_{Y | \Theta}(y | \Theta)  = \ex{ p_{Y | X,  \Theta}(y | X_1,\Theta)\,  p_{Y | X,  \Theta}(y | X_2, \Theta) \mid \Theta},
\]
where the expectation is with respect to independent vectors  $X_1$ and $X_2$ with the same distribution as $X$. By swapping the order of expectation between $\Theta$ and the pair $(X_1,X_2)$, we are then able to obtain closed form expressions for integrals involving the variance of the density. These expressions lead to explicit bounds with respect to the relative entropy (Theorem~\ref{thm:DKL_general}). 

Finally, the last step of our proof leverages Talagrand's transportation inequality \cite{talagrand:1996}, to obtain bounds on the conditional distribution of $Z$ given $\Theta$  with respect to the quadratic Wasserstein distance (Theorem~\ref{thm:W2_general}). This step requires careful control of the behavior of the conditional distribution $P_{Y|\Theta}$ in the limit as  the noise power $t$ converges to zero. 

One of the primary motivations for this work comes from the author's recent work on the asymptotic properties of a certain random linear estimation problem \cite{reeves:2016, reeves:2016a}. In particular, Theorem~\ref{thm:DKL_univariate} of this paper plays a key role in rigorously characterizing certain phase transitions that had been predicted using the heuristic replica method from statistical physics \cite{reeves:2012a}.  More generally, we believe that the results in this paper could be useful for the analysis of algorithms that rely on Gaussian approximations for weighted sums of large numbers of random variables. These include, for example, expectation propagation \cite{minka:2001}, expectation consistent approximate inference \cite{opper:2005},  relaxed belief propagation \cite{guo:2006},  and the rapidly growing class of algorithms based on approximate message passing \cite{donoho:2009a,bayati:2011, rangan:2011}. Another potential application for our results is to provide theoretical guarantees for approximate inference. Some initial work in this direction is described in \cite{boom:2015,boom:2015a}. 

\subsection{Statement of Main Results} 
Before we can state our main results we need some additional definitions. The quadratic Wasserstein distance between distributions $P$ and $Q$ on $\reals^k$ is defined according to
\begin{align*}
W_2(P,Q) & = \inf  \left( \ex{ \| U - V\|^2 }  \right)^\frac{1}{2},
\end{align*}
where the infimum is over all couplings of the random vectors $(U,V)$ obeying the marginal constraints $U \sim P$ and $V \sim Q$, and  $\|\cdot\|$ denotes the Euclidean norm. The quadratic Wasserstein distance metrizes the convergence of distributions with finite second moments; see e.g., \cite{villani:2003}. 

Another measure of the discrepancy between distributions $P$ and $Q$ is given by the relative entropy (also known as Kullback-Leibler divergence), which is defined according to 
\begin{align*}
\DKL{P}{Q} & = \int  \log\left( \frac{\dd P}{\dd Q}\right)  \dd P,
\end{align*}
provided that $P$ is absolutely continuous with respect to $Q$ and the integral exists.  Relative entropy is not a metric since it is not symmetric and does not obey the triangle inequality. Convergence with respect to relative entropy, which is sometimes referred to convergence in information, is much stronger than convergence in distribution \cite{csiszar:1984}.

The main results of this paper are bounds on the conditional distributions of the random projection $Z$ defined in \eqref{eq:Z_def} and the noisy random projection $Y$ defined in \eqref{eq:Y_def}. The marginal distributions of these vectors are denoted by $P_Z$ and $P_Y$ and the Gaussian distributions with the same mean and covariance are denoted by $G_Z$ and $G_Y$.   The conditional distributions corresponding to the random matrix $\Theta$ are denoted by $P_{Z|\Theta}$ and $P_{Y|\Theta}$. Using this notation, the marginal distributions can be expressed as $P_{Z} = \ex{ P_{ Z| \Theta}}$ and $P_{Y} = \ex{ P_{Y | \Theta}}$ where the expectation is with respect to $\Theta$.

The following definition describes the properties of the distribution of $X$ that are needed for our bounds. 

\begin{definition}
For any $n$-dimensional random vector $X$ with $\ex{ \|X\|^2} < \infty$, the functions $\alpha(X)$ and $\beta_r(X)$ are defined according to
\begin{align*}
\alpha(X) & = \frac{1}{n}  \ex{ \left| \|X\|^2 -  \ex{ \|X\|^2}  \right|} \\
\beta_r(X) &=\frac{1}{n}  \left( \ex{ \left| \langle X_1, X_2 \rangle  \right|^r} \right)^{\frac{1}{r}} ,
\end{align*}
where $r\in \{1,2\}$, $\langle \cdot, \cdot \rangle$ denotes the Euclidean inner product between vectors and $X_1$ and $X_2$ are independent vectors with the same distribution as $X$. 
\end{definition}

The function $\alpha(X)$ measures the deviation of the squared magnitude of $X$ about its expectation. The function $\beta_r(X)$ is non-decreasing in $r$. It is straightforward to show that the case $r = 2$ can be expressed equivalently as  $ \beta_2(X)   =    \left\| \frac{1}{n} \ex{X X^T} \right \|_F$, where $\|\cdot\|_F$ denotes the Frobenius norm.

\begin{assumption}[IID Gaussian Projections]\label{assumption:Gaussian_proj}The entries of the $k \times n$ matrix $\Theta$ are i.i.d.\  Gaussian random variables with mean zero and variance $1/n$. 
\end{assumption}

\begin{assumption}[Finite Second Moment]\label{assumption:finite_monent} The $n$-dimensional random vector $X$ has finite second moment: $\frac{1}{n} \ex{ \|X\|^2} = \gamma \in (0,\infty)$.
\end{assumption}

Under Assumptions 1 and 2, the marginal distribution of $Z$ has mean zero and covariance $ \gamma I_k$,  and thus the Gaussian approximations are given by $G_Z = \normal(0, \gamma I_k)$ and $  G_Y  = \normal\left(0, (\gamma + t) I_k\right)$. Furthermore, the functions $\alpha(X)$ and $\beta_2(X)$ satisfy:
\begin{align*}
0\le \frac{\alpha(X)}{\gamma} \le 2 \qquad \text{and} \qquad \frac{1}{\sqrt{n}} \le \frac{\beta_2(X)}{ \gamma} \le 1.
\end{align*}
The main results of the paper are given in the following theorems. 

\begin{theorem}\label{thm:W2_general} Under Assumptions 1 and 2, the quadratic Wasserstein distance between the conditional distribution of $Z$ given $\Theta$ and Gaussian distribution with the same mean and covariance as $Z$ satisfies 
\begin{align*}
\MoveEqLeft \frac{1}{\gamma} \ex{ W_2^2(P_{Z \mid \Theta},G_Z)}   \le C\,  k\,   \frac{ \alpha(X)}{\gamma}   \\
& \quad +  C\,   k^\frac{3}{4} \left(\frac{ \beta_1(X)}{\gamma}\right)^\frac{1}{2}   +   C\,  k  \left(\frac{\beta_2(X)}{ \gamma} \right)^\frac{4}{k+4} ,  
\end{align*}
where $C$ is a universal constant. In particular, the inequality holds with $C = 40$. 
\end{theorem}

\begin{theorem}\label{thm:DKL_general}
Under Assumptions 1 and 2, the relative entropy between the conditional distribution of $Y$ given $\Theta$ and Gaussian distribution with the same mean and covariance as $Y$ satisfies
\begin{align*}
\MoveEqLeft[1]  \ex{ \DKL{P_{Y \mid \Theta}}{G_Y} }  \le C \, k \,  \log\left(1 + \frac{\gamma}{t} \right)  \frac{\alpha(X)}{ \eps \gamma} \\
 &  \quad + C \,      k^\frac{3}{4}  \left(  \frac{\beta_1(X)}{\gamma}\right)^\frac{1}{2}   +  C\,  k^\frac{1}{4} \left(1 + \frac{ (2+\eps)\, \gamma }{t} \right)^\frac{k}{4} \frac{  \beta_2(X)}{  \gamma },
 \end{align*}
 for all $t \in (0,\infty)$ and $\eps \in (0,1]$ where $C$ is a universal constant. In particular, the inequality holds with $C =3$. 
\end{theorem} 

To interpret these results, it is useful to consider the setting where the functions $\alpha(X)$ and $\beta_2(X)$ are upper bounded by $C\, \gamma / \sqrt{n}$ for some fixed constant $C$.  This occurs, for example, when the entries of $X$ are independent with mean zero and finite fourth moments. 

\begin{cor}\label{cor:W2} Consider Assumptions 1 and 2. For any $n$-dimensional random vector $X$ satisfying
\begin{align*}
 \frac{\alpha(X)}{\gamma}  \le \frac{C}{ \sqrt{n}},  \qquad \frac{\beta_2(X)}{\gamma}  \le \frac{C } {\sqrt{n}},
\end{align*}
the quadratic Wasserstein distance satisfies
\begin{align*}
 \frac{1}{\gamma} \ex{ W_2^2(P_{Z \mid \Theta},G_Z)} &  \le C'\, \left( n^{-\frac{1}{4}} +  k \,   n^{-\frac{2}{k+4}} \right).
\end{align*}
\end{cor}
\begin{proof}
This result follows from combining Theorem~\ref{thm:W2_general} with the fact that $\beta_1(X) \le \beta_2(X)$, and then retaining only the dominant terms in the bound. 
\end{proof}

The proof of Theorem~\ref{thm:DKL_general} is given in Section~\ref{sec:KL_bounds}, which also provides some additional results. The proof of Theorem~\ref{thm:W2_general} is given in Section~\ref{sec:W2_bounds}. 

\subsection{Relation to Prior Work} 

We now compare our results to previous work in the literature.  The bounded-Lipschitz distance  between distributions $P$ and $Q$ on $\reals^k$ is defined according to
 \begin{align*}
d_\mathrm{BL}(P,Q)  & = \sup_{f}  \left | \int f\,  \dd P - \int f\,  \dd Q \right|, 
\end{align*}
where the supremum is over all functions $f : \reals^k \to [-1, 1]$ that are Lipschitz continuous with Lipschitz constant one. Convergence with respect to the bounded-Lipschitz distance is equivalent to convergence in distribution (also known as weak convergence). 

One of the central questions in the literature has been to provide conditions under which the conditional distribution of $Z$ given $\Theta$ converges to a Gaussian distribution weakly in probability for a sequence of problems indexed by the vector length $n$. Formally, this can be stated as
\begin{align}
\lim_{n\to \infty} \pr{ d_\mathrm{BL}(P_{Z| \Theta}, G_Z) > \eps}  =  0 \quad \text{for all $\eps > 0$},  \label{eq:converge_prob} 
\end{align} 
where the probability is with respect to the random matrix $\Theta$. 

When the number of projections $k$ is fixed, D\"umbgen and Zerial \cite{dumbgen:2013} show that  a necessary and sufficient condition for \eqref{eq:converge_prob} is given by
\begin{align}
\alpha(X) \to 0 \quad \text{and} \quad \beta_2(X) \to 0 \quad \text{as $n\to\infty$}.   \label{eq:cond_converge}
\end{align}
Strictly speaking,  \cite[Theorem~2.1]{dumbgen:2013} is stated in terms of convergence in probability, whereas $\alpha(X)$ and $\beta_r(X)$ correspond to expectations. However, under the assumption that $X$ has finite second moment, these conditions are equivalent. 

The sufficiency of \eqref{eq:cond_converge} can also be seen as a consequence of Theorem~\ref{thm:W2_general} and the fact that convergence with respect to the Wasserstein metric implies convergence in distribution. Moreover, the fact that \eqref{eq:cond_converge} is a necessary condition means that the dependence of our analysis on $\alpha(X)$ and $\beta_2(X)$ is optimal in the sense than any result bounding convergence in distribution must depend on these quantities. 

Another problem of interest is to characterize conditions under which \eqref{eq:converge_prob} holds in the setting where the number of projections increases with the vector length.  In this direction, Meckes \cite[Theorem~3]{meckes:2012} provides explicit bounds  with respect to the bounded-Lipchitz metric. Under the assumptions $\alpha(X) \le C / \sqrt{n}$ and $\lambda_\mathrm{max} \left( \ex{  X X^T} \right) \le C$ for some fixed constant $C$, Meckes shows that \eqref{eq:converge_prob} holds in the limit as both $k$ and $n$ increase to infinity provided that
\begin{align}
k  \le  \frac{ \delta \log n}{ \log \log n},   \label{eq:cond_converge2}
\end{align}
for some $\delta \in [0,2)$. Meckes also shows that this scaling is sharp in the the sense that  if $k = \delta \log n / \log \log n$ for some $\delta > 2$, then there exists a sequence of distributions for which \eqref{eq:converge_prob} does not hold. 

For comparison with the results in this paper, observe that that the function $\beta_2(X)$ satisfies 
\begin{align*}
\beta_2(X) & = \frac{1}{n} \sqrt{ \sum_{i=1}^n \lambda^2_i (\ex{ XX^T})}  \le \frac{1}{\sqrt{n}} \lambda_\mathrm{max}(\ex{ XX^T}),
\end{align*}
where equality is attained if and only if $\ex{XX^T}$ is proportional to the identity matrix.  Therefore, the condition on $\beta_2(X)$ in Corollary~\ref{cor:W2} is satisfied whenever $\lambda_\mathrm{max}(\ex{ XX^T}) \le C$. It is easy to verify that the scaling conditions under which the bound in Corollary~\ref{cor:W2} converges to zero are  the same as the conditions  given by Meckes.  As a consequence, we see that the scaling behavior of our results cannot be improved in general. Furthermore, we note that there can exist cases where the maximum eigenvalue $\lambda_\mathrm{max}(\ex{ XX^T})$ increases with the problem dimension while  $\beta_2(X)  \le C/\sqrt{n}$. In these cases,  the scaling conditions implied by our results are stronger than the ones provided by Meckes.  

A number of results in the literature have provided improved rates of convergence under further assumptions on $X$.  For example, Antilla, Ball, and Perissinaki \cite{anttila:2003} provide convergence rates when  $X$ is distributed uniformly on a centrally symmetric convex body, and  Bobkov  \cite{bobkov:2003} provides convergence rates when $X$ has a log-concave distribution.  The results in this paper are also related to the work of Hall and Li \cite{hall:1993} and Leeb \cite{leeb:2013}, who focus on certain properties of the bivariate distribution between two different projections. Finally, our bounds with respect to relative entropy are similar in spirit to work on entropic central limit theorems \cite{barron:1986aa, artstein:2004aa, madiman:2007aa, bobkov:2013,bobkov:2013b, bobkov:2014aa}. In particular, Bobkov \cite{bobkov:2013b}
combines entropic bounds with transportation inequalities to obtain bounds with respect to the Wasserstein metric.

\subsection{Some Consequences of our Results}

Many of the ideas behind our approach come directly from the author's recent work on the asymptotic properties of a certain random linear estimation problem \cite{reeves:2016, reeves:2016a}. For this problem, the behavior of the asymptotic mutual information and minimum mean-square error has been analyzed using the powerful but heuristic replica method from statistical physics \cite{reeves:2012a}. The main result in \cite{reeves:2016, reeves:2016a} is a rigorous proof that results obtained using the replica method are correct. One of the key steps is to study the behavior of the problem as the number of observations increases.  To this end, Theorem~\ref{thm:DKL_univariate} of this paper is used to characterize the conditional distribution of the new observation in terms of simple properties of the posterior distribution induced by the previous observations.

 Another application for the results in this paper is to understand the connections between information-theoretically optimal methods for signal acquisition and the framework of  compressed sensing  \cite{bresler:1999,donoho:2006a,candes:2006}, which seeks to recover an unknown vector from a small number of noisy linear projections. An interesting phenomenon in compressed sensing is that random projections have a certain universality property: projections chosen uniformly at random (subject to a power constraint) are often nearly as a good as projections that are designed optimally based on specific properties of the problem.  This phenomenon has been understood, to some extent, via connections with high-dimensional convex geometry, and in particular to the almost spherical property of low-dimensional sections of convex bodies, as described by Dvoretzky's theorem  (see e.g., \cite{vershynin:2014}). 

Using the results in this paper, we can obtain a more direct explanation for the universality of random projections in compressed sensing. Recall that the capacity of the AWGN channel with signal-to-noise ratio $s$ is given by $\cC(s) = \frac{1}{2} \log(1 + s)$ nats per channel use. The capacity provides an upper bound on the mutual information between the unknown vector and the observations generated according  to the optimal source and channel coding scheme. Meanwhile, the mutual information between the vector $X$ and the noisy linear projections $Y$ described in \eqref{eq:Y_def} corresponds directly to the mutual information that arises in compressed sensing with an i.i.d.\ Gaussian matrix. Interestingly, the gap between the capacity of the AWGN channel and the mutual information $I(X; Y  | \Theta)$ can be related directly to the relative entropy between the conditional distribution $P_{Y|\Theta}$ and the Gaussian approximation $G_Y$ via the following identity:
  \begin{align}
\ex{ \DKL{P_{Y\mid \Theta}}{G_Y}} = k\, \cC(\gamma / t) -  I(X; Y \gmid \Theta). \label{eq:CS_gap}
 \end{align}
The proof of this result follows straightforwardly from the decomposition \eqref{eq:EDKL_alt} given below and properties of differential entropy (see e.g., \cite[Chapter~8.6]{cover:2006}). In words, Identity~\eqref{eq:CS_gap} shows that the expected relative entropy considered in Theorem~\ref{thm:DKL_general} is precisely the difference between the upper bound on the mutual information of the optimal sensing function and the mutual information using an i.i.d.\ Gaussian matrix.  Consequently, whenever this term is small, one can conclude that compressed sensing with a random matrix is near optimal in terms of mutual information.

\subsection{Notation} 
We use $C$ to denote an absolute constant. In all cases, $C$ is positive and finite although the value may change from place to place.  The Euclidean norm is denoted by $\|\cdot\|$. The indicator function of a set $\set$ is denoted by $\one_\set(\cdot)$. The positive part of a number $x$ is denoted by $(x)_+  = \max(x,0)$. All logarithms are stated with respect to the natural base. The multivariate Gaussian distribution with mean $\mu$ and covariance $\Sigma$ is denoted by $\normal(\mu, \Sigma)$. The joint distribution of random variables $X,Y$ is denoted by $P_{X,Y}$ and the marginal distributions are denoted by $P_X$ and $P_Y$. The conditional distribution of $X$ given $Y=y$ is denoted by $P_{X\mid Y = y}$ and the conditional distribution corresponding to a random realization of $Y$ is denoted by $P_{X\mid Y}$.

\section{Bounds on Relative Entropy}\label{sec:KL_bounds} 

\subsection{Decomposition of Relative Entropy} 

The starting point our analysis is based on the following identity for relative entropy and mutual information \cite{topsoe:1967}:
\begin{align}
\bEx \left[ \DKL{P_{Y\mid \Theta} }{  G_{Y} } \right]  & = \DKL{P_{Y}  }{  G_{Y} }  +I( Y; \Theta).  \label{eq:EDKL_alt} 
\end{align}
In this decomposition, the relative entropy $ \DKL{P_{Y}  }{  G_{Y} }$ depends on the difference between the marginal distribution of $Y$ and the Gaussian distribution with the same mean and covariance whereas the mutual information $I(Y;\Theta)$ is a measure of the dependence between $Y$ and $\Theta$.

For the setting considered in this paper, $\DKL{P_{Y}  }{  G_{Y} }$ can be addressed straightforwardly using the further decomposition
\begin{align*}
\DKL{P_{Y}  }{  G_{Y} } & =  \ex{ \DKL{P_{Y |X }  }{  G_{Y}} }  -  I(Y; X) \\
& = \frac{k}{2} \ex{ \log\left( \frac{t + \gamma}{ t + \frac{1}{n} \|X\|^2}  \right)}  -  I(Y; X) ,
\end{align*}
where the second line follows from the fact that the conditional distribution of $Y$ given $X$ is Gaussian with mean zero and covariance $(t + \frac{1}{n} \|X\|^2) I_n$. The first term on the right-hand side is a measure of the deviation of the squared magnitude of $X$ about its expectation. Lemma~\ref{lem:log_dev_inq} in the appendix gives
\begin{align}
\DKL{P_{Y}  }{  G_{Y} } & \le \frac{k}{2} \log\left( 1 + \frac{\gamma}{t} \right)  \frac{ \alpha_1(X)}{ \gamma} . \label{eq:DPY_to_alpha}
\end{align}

\subsection{Mutual Information Inequalities}

Our approach to bounding the mutual information $I(Y;\Theta)$ is based on certain integrals involving the variance of the conditional density of $Y$ given $\Theta$. Let $p_{Y}(y)$ and $p_{Y|\Theta}(y|\theta)$ denote the density functions of $P_{Y}$ and $P_{Y|\Theta =\theta}$, respectively. Note that for fixed $y$ and random $\Theta$ the density $p_{Y|\Theta}(y |\Theta)$ is a random variable whose expectation is given by the marginal density $\ex{ p_{Y|\Theta}(y |\Theta)} = p_{Y}(y)$. The variance of the conditional density is a function from $\reals^k$ to $\reals_+$ that can be expressed as
\begin{align*}
\var(p_{Y|\Theta}(y | \Theta) )  = \ex{ \left( p_{Y|\Theta}(y|\Theta) - p_Y(y) \right)^2}.
\end{align*}

The variance of the conditional density provides a measure of the dependence between $Y$ and $\Theta$. The next result shows that the integral of the square root of the variance gives an upper bound on the mutual information.

\begin{lemma}\label{lem:MI_var_bound} The mutual information satisfies
\begin{align}
I(Y; \Theta) & \le  \kappa \,   \int_{\reals^k} \sqrt{ \var(p_{Y \mid \Theta}(y \gmid \Theta)) } \,  \dd y ,  \notag 
\end{align}
where
\begin{align}
\kappa =  \sup_{x \in (0,\infty) } \log(1 + x)/\sqrt{x}  \approx  0.80474  \label{eq:kappa}.
\end{align}
\end{lemma}
\begin{proof}
The chi-squared distance between distributions $P$ and $Q$ with densities $p$ and $q$ with respect to a dominating measure $\mu$ is defined by $\chi^2(P,Q) = \int ( \frac{p}{q} -1)^2   q\,   \dd \mu$. The chi-square dominates the  relative entropy and satisfies  $\DKL{P}{Q} \le \log(1 + \chi^2(P,Q)) \le \kappa \sqrt{ \chi^2(P,Q)}$ where the first inequality is from \cite[Theorem~5]{gibbs:2002} and the second inequality follows from the definition of $\kappa$. Therefore, the mutual information can be upper bounded using 
\begin{align}
I(Y;\Theta) & = \int p_Y(y)\,  \DKL{P_{\Theta|Y=y} }{P_\Theta} \dd y \notag\\
&  \le   \kappa \int p_Y(y) \sqrt{ \chi^2(P_{\Theta|Y=y}, P_{\Theta}) }  \,  \dd y. \label{eq:MI_var_bound_b} 
\end{align} 
The chi-squared distance can be related to the variance of the conditional density by noting that
\begin{align*}
 \chi^2(P_{\Theta|Y=y}, P_{\Theta}) & = \ex{ \left( \frac{p_{Y|\Theta}(y |\Theta)}{p_Y(y)} - 1 \right)^2 }\\
 & = \frac{ \var(p_{Y|\Theta}(y | \Theta) ) }{p^2_Y(y)},
\end{align*}
where the first equality follows from Bayes' rule. Plugging this identity back into \eqref{eq:MI_var_bound_b} completes the proof. 
\end{proof}


Our next result is a general inequality that allows us to bound the integral of the square root of a function in terms of certain moments.  The $p$-th moment of an integrable function $f: \reals^k \to \reals$ is defined according to
\begin{align*}
\mu_p(f) & = \int \|y\|^p f(y) \, \dd y.
\end{align*}

\begin{lemma}\label{lem:int_moment_inq}
For any non-negative integrable function $f : \reals^k \to \reals_+$  with $\mu_{k-1}(f), \mu_{k+1}(f) < \infty$,
\begin{align}
\int \sqrt{ f(y)  }\, \dd y & \le \sqrt{ \frac{  2 \pi^{\frac{k}{2} + 1}}{ \Gamma(\frac{k}{2} )}}  \big( \mu_{k-1}(f)\,  \mu_{k+1}(f) \big)^\frac{1}{4} , \notag 
\end{align}
where $\Gamma(z) = \int_0^\infty x^{z-1} e^{-x} \, \dd x$ is the Gamma function. 
\end{lemma}
\begin{proof}
Let $g(y) = \left( \lambda \|y\|^{k-1} + \lambda^{-1} \|y\|^{k+1} \right)^{-1}$ where $\lambda = \sqrt{\mu_{k+1}(f) / \mu_{k-1}(f)}$. Using the Cauchy-Schwarz inequality, we have
\begin{align}
\int \sqrt{ f(y)  }\, \dd y & = \int \sqrt{ g(y)}  \sqrt{\frac{ f(y)  }{g(y)} }\, \dd y \notag \\
& \le \sqrt{ \int g(y) \, \dd y} \sqrt{ \int \frac{f(y)}{g(y)}\,   \dd y}. \label{eq:lem:int_moment_inq_b}
\end{align}
Letting $\omega_k = \pi^\frac{k}{2}/ \Gamma( \frac{k}{2} + 1)$ denote the volume of the $k$-dimensional Euclidean ball, the first integral can be computed directly as
\begin{align}
\int g(y)\, \dd y 
&  =   \int_0^\infty  \frac{k\,  \omega_n}{ 1 + u^2}  \, \dd u  =k \,  \omega_k \, \frac{ \pi}{2}  \label{eq:lem:int_moment_inq_c},
\end{align}
where the first step follows from a transformation to polar coordinates.  Meanwhile, the second integral is given by
\begin{align}
 \int \frac{f(y)}{g(y)}  \dd y & = \lambda\,  \mu_{k-1}(f) + \lambda^{-1} \mu_{k+1}(f) \notag \\
 &  = 2 \sqrt{ \mu_{k-1}(f)\,  \mu_{k+1}(f)}. \label{eq:lem:int_moment_inq_d}
\end{align}
Plugging \eqref{eq:lem:int_moment_inq_c} and \eqref{eq:lem:int_moment_inq_d} back into \eqref{eq:lem:int_moment_inq_b}  leads to the stated inequality.
\end{proof}

To proceed we introduce the following definitions: 
\begin{align}
m_{p}(Y,\Theta) 
& = \frac{ \int \|y\|^p \var(p_{Y \mid \Theta}(y  \gmid \Theta)) \, \dd  y}{ \int \|y\|^p \phi^2(y) \, \dd y} ,\notag \\
M(Y,\Theta) & = \sqrt{ m_{k-1}(Y,\Theta)\, m_{k+1}(Y,\Theta)}, \notag
\end{align}
where $\phi(y) = (2\pi)^{-\frac{k}{2}} \exp( -\frac{1}{2} \|y\|^2)$ is the standard Gaussian density on $\reals^k$.  The next result follows from Lemma~\ref{lem:MI_var_bound} and Lemma~\ref{lem:int_moment_inq}.

\begin{lemma}\label{lem:MI_to_M_kp}  For any random pair $(Y,\Theta)$ with $M(Y,\Theta)  < \infty$, the mutual information satisfies 
\begin{align}
I(Y; \Theta) & \le  \kappa\,    \left(  \frac{  \pi k }{ 2}    \right)^\frac{1}{4} \sqrt{ M(Y,\Theta ) },  \notag 
\end{align}
where $\kappa$ is defined in \eqref{eq:kappa}. 
\end{lemma}
\begin{proof}
The normalization term in the definition of $m_p(Y,\Theta)$ can be computed explicitly as $\mu_p(\phi) = (4 \pi)^{-\frac{k}{2}} \Gamma( \frac{k+p}{2} ) /\Gamma(\frac{k}{2})$. Combining Lemma~\ref{lem:MI_var_bound} and Lemma~\ref{lem:int_moment_inq} leads to
\begin{align*}
I(Y; \Theta) & \le  \kappa\,   \sqrt{ \rho(k) \, M(Y,\Theta) } ,  
\end{align*}
where $\rho(k) = \pi 2^{1-k} \left[\Gamma\left(k- \frac{1}{2}\right) \Gamma\left( k + \frac{ 1}{2}\right) \right]^\frac{1}{2} / \Gamma^2(\frac{k}{2})$. This is a slightly stronger version of the stated inequality. 
Using the Legendre duplication formula for the Gamma function \cite[Equation~(1.7)]{qi:2010}, the function $\rho(k)$ can be expressed as
\begin{align*}
\rho(k) = \sqrt{ \frac{ \pi k}{2}} \frac{ \xi\left(\frac{k}{2}\right)}{ \xi\left(k - \frac{1}{2}\right)},
\end{align*} 
with  $\xi(z) =z^{-\frac{1}{2}} \Gamma(z + \frac{1}{2} ) /  \Gamma(z)$.  The function $\xi(z)$ is non-decreasing on the positive reals \cite[Section~3.16]{qi:2010}, and thus we can conclude that $\rho(k) \le \sqrt{ \pi k / 2}$ for all $k \ge 1$.
\end{proof}

\begin{remark}
Lemma~\ref{lem:MI_to_M_kp}   holds generally  for any random pair $(Y, \Theta)$ such that the conditional distribution of $Y$ given $\Theta$ is absolutely continuous with respect to Lebesgue measure on $\reals^k$. 
\end{remark}
\subsection{Characterization of Moments}\label{sec:char_of_moments} 

The next step in our analysis is to characterize the moments of the variance of the conditional density.   Let  $X_1$ and $X_2$ be independent copies of $X$ and let the random tuple $(\Va, \Vg, R)$ be defined according to
\begin{align*}
\Va &= t+ \frac{ 1}{2 n} \|X_1\|^2 + \frac{1}{2n} \|X_2\|^2\\
\Vg &= \sqrt{ \left(t+ \frac{ 1}{ n} \|X_1\|^2\right)\left(t+ \frac{1}{n} \|X_2\|^2\right)} \\
R  & =  \frac{ 1}{n} \langle X_1, X_2 \rangle.
\end{align*}
The variables $\Va$ and $\Vg$ correspond to the arithmetic and geometric means respectively of $\{ t + \frac{1}{n} \|X_i\|^2\}_{i \in \{1,2\}}$, and thus $0 \le \Vg \le \Va$.  The variables $\Va$ and $R$ can be related to the the sum and the difference of $X_1$ and $X_2$ using the identities $\Va + R =t+  \frac{1}{2n} \|X_1 + X_2\|^2$ and $\Va - R = t+  \frac{1}{2n} \|X_1 - X_2\|^2$.

The next result gives an explicit characterization of $m_{p}(Y,\Theta)$ in terms of an expectation with respect to the tuple $(\Va, \Vg, R)$. 
\begin{lemma}\label{lem:M_kp_characterization} If $k + p > 0$ and  $\ex{\|X\|^\frac{p}{2} } < \infty$ then $m_{p}(Y, \Theta)$ is finite and is given by
\begin{align*}
\MoveEqLeft m_{p}(Y,\Theta)=\\
&\ex{   \left( \frac{1}{\Va  - R} \right)^\frac{k}{2} \bigg( \frac{ \Vg^2 - R^2}{ \Va- R} \bigg)^\frac{p}{2} -   \left( \frac{1}{\Va } \right)^\frac{k}{2} \bigg( \frac{ \Vg^2}{ \Va} \bigg)^\frac{p}{2}   }. 
\end{align*}
\end{lemma}
\begin{proof}
We begin by noting that the conditional density can be expressed as $p_{Y|\Theta}(y|\theta)  = \ex{ \phi_t(y - \theta X)}$, where $\phi_t(y) = (2 \pi t)^{-\frac{k}{2}} \exp( -\frac{1}{2t} \|y\|^2)$ is the Gaussian density on $\reals^k$ with mean zero and covariance $t I_k$. The key idea of the proof is to use the fact that the square of the conditional density can be expressed as
\begin{align}
p^2_{Y|\Theta}(y|\theta) & = \ex{ \phi_t(y - \theta X_1)  \phi_t(y - \theta X_2)},   \label{eq:squared_density_alt} 
\end{align}
where $X_1$ and $X_2$ are independent copies of $X$. Then, taking the expectation of both sides with respect to a random matrix $\Theta$, and then swapping the order of the expectation with respect to $\Theta$ and $(X_1,X_2)$ allows us to write
\begin{align*}
\ex{ p^2_{Y|\Theta}(y|\Theta)}  & = \ex{\nu(y, X_1, X_2)},
\end{align*}
where $\nu(y,x_1,x_2)  = \ex{  \phi_t(y - \Theta x_1)  \phi_t(y - \Theta x_2)}.$

The next step is to obtain a simplified expression for  $\nu(y,x_1,x_2)$. Observe that for any fixed pair $(x_1,x_2)$ and random matrix $\Theta$, the vectors $\Theta x_1$ and $\Theta x_2$ are jointly Gaussian with 
\begin{align*}
\begin{bmatrix} \Theta x_1 \\ \Theta x_2 \end{bmatrix}  \sim \normal(0, \Sigma ), \quad \Sigma= \frac{1}{n}  \begin{bmatrix} \|x_1\|^2 & \langle x_1, x_2 \rangle \\ \langle x_1, x_2 \rangle & \|x_2\|^2 \end{bmatrix}  \otimes I_k.
\end{align*} 
As a consequence, the expectation with respect to $\Theta$ in the definition of $\nu(y,x_1,x_2)$ can be expressed as a function of $\Sigma$ using
\begin{align*}
\nu(y,x_1,x_2) & =(2\pi t)^{- k}  \ex{ \exp\left( - \frac{1}{2 t} \left\|  \begin{bmatrix} y  \\ y \end{bmatrix} -\begin{bmatrix}  \Theta x_1 \\  \Theta x_2 \end{bmatrix}  \right\|^2   \right) }\\
& = (2\pi)^k  \left( \det (\Sigma + t I) \right)^{-\frac{1}{2}}\\
& \quad \times  \exp\left( - \frac{1}{2} \left\|( \Sigma + tI )^{-\frac{1}{2}}   \begin{bmatrix} y \\ y \end{bmatrix} \right\|^2    \right),
\end{align*}
where the second step follows from recognizing the expectation as the moment generating function of a noncentral Wishart matrix \cite[Theorem~3.5.3]{gupta:1999}. After some straightforward algebra, we see that
\begin{align*}
\det ( \Sigma  + tI )    &  = \left( v_g^2 - r^2\right)^k  \\
\frac{1}{2 } \left\|(  \Sigma + tI )^{-\frac{1}{2}}   \begin{bmatrix} y \\ y \end{bmatrix} \right\|^2 & =  \left( \frac{v_a - r}{v_g^2 - r^2}  \right)  \|y\|^2 ,
\end{align*}
where $(v_a,v_g,r)$ are defined in the same way as $(V_a, V_g, R)$.

Using results given above, the term inside the expectation in \eqref{eq:squared_density_alt}  can be expressed as
\begin{align*}
\nu(y, X_1, X_2 )  
& = (V_a - R)^{-\frac{k}{2}} U^\frac{k}{2}  \phi^2\left(U^{-\frac{1}{2}}  y  \right),
\end{align*}
where $U  = (\Vg^2 - R^2) / (V_a - R)$ and $\phi(y) = \phi_1(y)$. Using this representation, the $p$-th moment of $\nu(y, X_1, X_2 )$ with respect to $y$ can be expressed as
\begin{align*}
\MoveEqLeft  \int  \|y\|^p\nu(y, X_1, X_2 )  \, \dd y \\
& = (V_a - R)^{-\frac{k}{2}}  \int  \|y\|^pU^\frac{k}{2}  \phi^2\left(U^{-\frac{1}{2}}  y  \right) 
  \, \dd y \\
  & =  (V_a - R)^{-\frac{k}{2}} U^\frac{p}{2}   \int  \|z\|^p   \phi^2\left( z  \right)   \, \dd z,
\end{align*}
where the last step follows from the change of variables $z= U^{-\frac{1}{2}} y$. Taking the expectation of both sides and dividing by $ \mu_p(\phi^2) $ gives
\begin{align*}
\frac{ \mu_p\left(\ex{ p^2_{Y|\Theta}(y \gmid \Theta)}\right) }{ \mu_p(\phi^2)} = \ex{  \left(\frac{1}{V_a - R}\right)^\frac{k}{2} \left( \frac{\Vg^2 - R^2}{V_a - R} \right)^\frac{p}{2}}.
\end{align*}

In order to complete the proof, we also need to compute the $p$-th moment of $p_{Y}^2(y)$. We use the representation 
\begin{align*}
p_Y^2(y) = \ex{ \tilde{\nu}(y,x_1,x_2)},
\end{align*}
where $\tilde{\nu}(y) = \ex{  \phi_t(y - \Theta_1 x_1)  \phi_t(y - \Theta_2 x_2)}$  and $\Theta_1$ and $\Theta_2$  are independent copies of $\Theta$.  From here, we follow the same steps as before, with the main difference being that $\Theta_1 x_1$ and $\Theta_2 x_2$ are uncorrelated, that is
\begin{align*}
\begin{bmatrix} \Theta x_1 \\ \Theta x_2 \end{bmatrix}  \sim \normal(0, \widetilde{\Sigma }), \quad \widetilde{ \Sigma}= \frac{1}{n}  \begin{bmatrix} \|x_1\|^2 &  0  \\ 0  & \|x_2\|^2 \end{bmatrix}  \otimes I_k.
\end{align*} 
The resulting characterization of the $p$-th moment is given by
\begin{align}
\frac{ \mu_p\left(\ex{ p^2_{Y}(y )}\right) }{ \mu_p(\phi^2)} = \ex{  \left(\frac{1}{V_a }\right)^\frac{k}{2} \left( \frac{\Vg^2}{V_a } \right)^\frac{p}{2}}.
\end{align}
This completes the proof. 
\end{proof}

In some cases, the characterization of $m_{p}(Y,\Theta)$ given in Lemma~\ref{lem:M_kp_characterization} can be computed explicitly. 

\begin{example}[Orthogonal Support]
Suppose that $X$ is distributed on a set of $d\le n$ orthogonal vectors $\{x_1, \cdots, x_d\}$ with $\|x_i\|^2 = \gamma \, n$. Then $\Vg = \Va = t + \gamma$ and the distribution on $R$ is given by
\begin{align*}
R = \begin{cases}
\gamma, & \text{with probability $\lambda$}\\
0 , & \text{with probability $1-\lambda$} ,
\end{cases}
\end{align*}
where $\lambda = \sum_{i=1}^d P^2_{X}(\{x_i\}) = \Pr( X_1 = X_2)$.  By Lemma~\ref{lem:M_kp_characterization}, this means that
\begin{align*}
m_{p}(Y,\Theta)  & =\lambda \,  \left( \frac{ (t + 2 \gamma)^\frac{p}{2} }{ t^\frac{k}{2} } - \frac{(t+\gamma)^\frac{p}{2} }{(t+\gamma)^\frac{k}{2} }\right).
\end{align*}
Note that $\lambda \ge 1/d$, with equality when $X$ is distributed uniformly. 
\end{example}

\begin{example}[Uniform on Sphere]
Suppose that $X$ is uniform on the Euclidean sphere of radius $\sqrt{n\, \gamma}$. Then, it can be shown that
\begin{align*}
m_{p}(Y,\Theta)  & = \left(  \ex{ \frac{ (t + \gamma( 1 + U) )^\frac{p}{2} }{ \left( t + \gamma (1 - U)\right)^\frac{k}{2} }} - \frac{(t+\gamma)^\frac{p}{2} }{(t+\gamma)^\frac{k}{2} }\right),
\end{align*}
where $U$ is symmetric about zero with $U^2 \sim \text{Beta}(1, n-1)$. 
In this case, it is interesting to note that if $p$ is sufficiently small relative to $k$, then the function $M_{k,p}(X,t)$ is bounded uniformly with respect to $t$. 
\end{example}

The next results provide bounds on  $M(Y,\Theta)$ in terms of the functions $\alpha_r(X)$ and $\beta_r(X)$. The proofs of these results along with some further bounds are given in Appendix~\ref{sec:prop_M}. The first bound provides a general inequality for one-dimensional projections. The second bound applies to any distribution with bounded magnitudes. 

\begin{lemma}\label{lem:M_k1_bound} 
If $k=1$ and  $\ex{ \|X\|^2} < \infty$ then 
\begin{align*}
M(Y,\Theta) & \le \frac{\beta_1(X)}{ t} .
\end{align*}
\end{lemma}

\begin{lemma}\label{lem:M_kq_bound}
If $\Pmin \le \frac{1}{n} \|X\|^2 \le \Pmax$ almost surely, then 
\begin{align*}
 M(Y,\Theta) &\le \left( \frac{2 \Pmax}{\Pmin} \right)^\frac{1}{4} \\
 & \quad \times  \left[ k\,  \frac{ \beta_1(X)}{\Pmin}  + \left(1 + \frac{ 2 \Pmax}{t} \right)^\frac{k}{2} \frac{\beta_2^2(X)}{ \Pmin^2}  \right].
\end{align*}
Furthermore, if $\frac{1}{n} \|X\|^2 = \gamma$ almost surely, then 
\begin{align*}
 M(Y,\Theta) & \le 2^\frac{1}{4}    \left[k\,   \frac{\frac{1}{n} \ex{ \|X\|^2}}{\gamma}  +\left(1 + \frac{ 2 \gamma}{t} \right)^\frac{k}{2}  \frac{\beta_2^2(X)}{ \gamma^2}  \right].
\end{align*}
\end{lemma}

\subsection{Further Results and Proof of Theorem~\ref{thm:DKL_general}  } \label{sec:proofs_of_thm_DKL}

Using the results given in the previous section, we are now ready to give bounds on the relative entropy in terms of the parameters $\alpha_r(X)$ and $\beta_r(X)$. We begin with some special cases. The next result corresponds to the case of a one-dimensional projection. 

\begin{theorem}\label{thm:DKL_univariate} Consider Assumptions 1 and 2.  If $k=1$, then the relative entropy satisfies
\begin{align*}
 \ex{ \DKL{P_{Y \mid \Theta}}{G_Y} } & \le \frac{ \alpha_1(X)}{ 2 t}  +   \sqrt{ \frac{ \beta_1(X)}{t} } .
\end{align*}
\end{theorem} 
\begin{proof}
From Lemma~\ref{lem:MI_to_M_kp} and Lemma~\ref{lem:M_k1_bound} we see that  $I(Y;\Theta) \le \sqrt{ \beta_1(X)/t}$. Combining this inequality with \eqref{eq:EDKL_alt} and \eqref{eq:DPY_to_alpha} leads to the stated result.  
\end{proof}

\begin{theorem}\label{thm:DKL_sphere} Consider Assumption 1. If $\frac{1}{n} \|X\|^2 = \gamma$ almost surely, then
\begin{align*}
\MoveEqLeft \ex{ \DKL{P_{Y \mid \Theta}}{G_Y} }  \\
 & \le   k^\frac{3}{4} \left( \frac{   \frac{1}{n} \|\ex{X}\|^2 }{\gamma } \right)^\frac{1}{2}  + k^\frac{1}{4} \left(1 + \frac{ 2 \gamma }{t} \right)^\frac{k}{4} \frac{ \beta_2(X)}{\gamma} .
 \end{align*}
\end{theorem} 
\begin{proof}
The mutual information $I(Y;\Theta)$ can be upper bounded using  Lemma~\ref{lem:MI_to_M_kp} and Lemma~\ref{lem:M_kq_bound}, and noting that the  constant  $\kappa \pi^\frac{1}{4}2^\frac{1}{8}$ is less than one. Combining this bound with \eqref{eq:EDKL_alt} and noting that $\DKL{P_Y }{G_Y} = 0$ completes the proof.   
\end{proof}

It is interesting to note that the first term in Theorem~\ref{thm:DKL_sphere} is the norm of the expected value of $X$, and thus this term is equal to zero whenever $X$ has zero mean. 

At this point, the difficulty in bounding the mutual information for large $k$ and general distributions on $X$ arises from the fact that the behavior of the moments $m_p(Y,\Theta)$ can be dominated by the tail behavior of $\|X\|$. In particular, the requirement of higher order moments for $X$ is highly restrictive. The next result provides a conditioning argument that allows us to bypass this issue. 

\begin{lemma}\label{lem:DKL_decomp} For every measurable subset $\set \subseteq \reals^n$, the mutual information satisfies 
\begin{align}
 I(Y ; \Theta) &\le  \frac{k}{2} \log\left(1 + \frac{\gamma}{t} \right)  \left( P_X(\set^c) + \frac{\alpha_1(X)}{\gamma} \right) \notag \\
& \quad +  I(Y ; \Theta \gmid X \in \set) P_X(\set),  \label{eq:DKL_decomp}
\end{align}
where $\set^c = \reals^n \backslash \set$. 
\end{lemma} 
\begin{proof}
Let $U = \one_\set(X)$ be an indicator of the event $\{ X \in \set\}$. By the chain rule for mutual information, we have
\begin{align*}
I(Y,U; \Theta) & = I(Y; \Theta ) + I(U; \Theta \gmid Y) \\
& = I(Y; \Theta \gmid U)  + I(U; \Theta) .
\end{align*}
The term $I(U;\Theta)$ is equal to zero because $U$ and $\Theta$ are independent, and  rearranging terms leads to
\begin{align}
I(Y;\Theta) 
 & \le   I(Y ; \Theta \gmid X \notin \set) P_{X}(\set^c) \notag \\
 & \quad    +   I(Y ; \Theta \gmid X \in \set) P_{X}(\set).  \label{eq:DKL_decomp_c}
\end{align}
The mutual information  in the first term on the right-hand side can be further bounded using
\begin{align}
I(Y ; \Theta \gmid X \notin \set)   & \le I(Y ; \Theta \gmid X, X \notin \set) \notag \\
& =\frac{k}{2}  \ex{  \log\left( \frac{ t + \frac{1}{n} \|X\|^2 }{ t  } \right) \; \middle \vert \; X \notin \set},  \label{eq:DKL_decomp_d}
\end{align}
where the inequality follows from the same steps that let to \eqref{eq:DKL_decomp_c}, and the second step follows from the fact that the conditional distribution of $Y$ given $(\Theta, X)$ is Gaussian with mean $\Theta X$ and covariance $t I_k$.  If we multiply this term by the probability $P_X(\set^c)$, we then have
\begin{align}
\MoveEqLeft[1.5] \ex{  \log\left( \frac{ t + \frac{1}{n} \|X\|^2 }{ t  } \right) \; \middle \vert \; X \notin \set} P_X(\set^c) \notag \\
& = \ex{  \log\left( \frac{ t + \frac{1}{n} \|X\|^2 }{ t +\gamma  } \right)  \one_{\set^c}(X)}+ \log\left(1 + \frac{\gamma}{t} \right)P_X(\set^c) \notag \\
& \le  \log\left( 1 + \frac{\gamma}{t} \right)  \left( \frac{ \alpha_1(X)}{\gamma}  + P_X(\set^c)\right),  \label{eq:DKL_decomp_e}
\end{align} 
where the inequality follows from Lemma~\ref{lem:log_dev_inq}. Combining \eqref{eq:DKL_decomp_c}, \eqref{eq:DKL_decomp_d}, and \eqref{eq:DKL_decomp_e} completes the proof. 
\end{proof}

We are now ready to prove Theorem~\ref{thm:DKL_general}. Given $\eps \in (0,1]$, let $\set$ be defined according to
\begin{align}
\set = \left\{ x \in \reals^n \, : \,  \left| \frac{1}{n} \|x\|^2  -\gamma \right|  \le \frac{ \eps}{2}  \gamma \right\}.
\end{align}
By Markov's inequality, the probability $P_X(\cE^c)$ can be upper bounded in terms of the deviation of the squared magnitude of $X$ about its expectation: 
\begin{align}
P_X(\cE^c) \le \frac{2}{\eps \gamma} \ex{ \left|  \frac{1}{n} \|X\|^2 - \gamma \right| } =  \frac{2}{ \eps} \frac{\alpha_1(X)}{ \gamma}. \label{eq:DKL_general_c}
\end{align}

Next, let $X'$ denote a vector that is drawn according to the conditional distribution of $X$ given $X \in \set$ and let $Y' = \Theta X' + \sqrt{t} N$ denote the corresponding measurements. By construction, the magnitude of $X'$ is bounded almost surely:
\begin{align*}
\tfrac{1}{2} \gamma \le  \left(1 -  \tfrac{\eps}{2}\right) \gamma  \le  \frac{1}{n} \|X'\|^2 \le \left(1  + \tfrac{\eps}{2}\right) \gamma \le \tfrac{3}{2} \gamma  .
\end{align*}
Therefore, by Lemma~\ref{lem:M_kq_bound},  we have
\begin{align*}
M(Y',\Theta)& \le \left( \frac{2 (1  +\eps)  }{(1-\eps) } \right)^\frac{1}{4}  \times  \\
& \quad   \left[ k  \frac{  \beta_1(X')}{(1-\eps) \gamma }  +  \left(1 + \frac{ 2(1+\eps) \gamma }{t} \right)^\frac{k}{2} \frac{\beta_2^2(X')}{ (1-\eps)^2 \gamma^2 } \right]\\
& \le2^\frac{9}{4}  3^\frac{1}{4}   \left[  k  \frac{ \beta_1(X')}{\gamma}  + \left(1 + \frac{ (2+\eps) \gamma }{t} \right)^\frac{k}{2} \frac{\beta_2^2(X')}{  \gamma^2 }  \right].
\end{align*}
The function $\beta_r(X')$ can be related to $\beta_r(X)$ by noting that
\begin{align*}
\beta_r^r(X') P^2_X(\set)  & = \frac{1}{n^r} \ex{ \left| \langle X_1 , X_2 \rangle \right|^r \one_\set(X_1) \one_\set(X_2)}\\
& =  \beta_r^r(X \one_\set(X))\\
& \le  \beta_r^r(X).
\end{align*}
Combining these inequalities with Lemma~\ref{lem:MI_to_M_kp} and Lemma~\ref{lem:M_k1_bound} leads to
\begin{align}
\MoveEqLeft I(Y ; \Theta \gmid X \in \set) P_X(\set) \le 2^\frac{9}{4}  3^\frac{1}{4}  \notag \\
& \quad \times  \left[  k  \frac{ \beta_1(X)}{\gamma}  + \left(1 + \frac{ (2+\eps) \gamma }{t} \right)^\frac{k}{2} \frac{\beta_2^2(X)}{  \gamma^2 }  \right]^\frac{1}{2} . \label{eq:DKL_general_d}
\end{align}
Finally, the proof of Theorem~\ref{thm:DKL_general} is completed by combining \eqref{eq:EDKL_alt}, \eqref{eq:DPY_to_alpha},  and Lemma~\ref{lem:DKL_decomp} with \eqref{eq:DKL_general_c} and \eqref{eq:DKL_general_d}.

\section{Bounds on Wasserstein Distance} \label{sec:W2_bounds}

This section provides bounds with respect to the expected squared Wasserstein distance of order two. Our first result  follows from Talagrand's transportation inequality \cite{talagrand:1996} and shows that the Wasserstein distance can be upper bounded in terms of the relative entropy between the distribution $P_{Y \mid \Theta}$ and $G_Y$ defined in Section~\ref{sec:KL_bounds}. 

\begin{lemma}\label{lem:W2_to_KL}The Wasserstein distance satisfies the following inequality for every realization of the matrix $\Theta$
\begin{align*}
W_2^2(P_{Z \mid \Theta}, G_Z) & \le 4 t k  +  4 (t+ \gamma) \DKL{P_{Y \mid \Theta}}{G_Y}. 
\end{align*}
\end{lemma}
\begin{proof}
Two applications of the triangle inequality yields:
\begin{align*}
W_2(P_{Z \mid \Theta} , G_Z) & \le W_2(P_{Z \mid \Theta} , P_{Y \mid \Theta})  + W_2( P_{Y \mid \Theta}, G_Y) \\
& \quad  + W_2(G_Y , G_Z),
\end{align*}
where  $P_{Y \mid \Theta} = P_{Z \mid \Theta} \ast \normal(0, tI_k)$ and $G_{Y} = G_{Z } \ast \normal(0, tI_k)$. 
By the subadditivity of Wasserstein distance under convolution  \cite[Proposition 7.17]{villani:2003}, it follows that both $W_2(P_{Z \mid \Theta} , P_{Y \mid \Theta})$ and $W_2(G_Y , G_Z)$ are upper bounded by $\sqrt{ t k}$. Combining these bounds with the inequality $(a+b)^2 \le 2 a^2 + 2b^2$ leads to
\begin{align}
W_2^2(P_{Z \mid \Theta} , G_Z) & \le 4 tk   +2 W_2^2(P_{Y \mid \Theta}, G_Y).   \label{eq:W2_to_KL_b}
\end{align}

Talagrand's transportation inequality \cite{talagrand:1996} gives $W^2_2(Q , G_Y)  \le  2 k  ( t + \gamma)  \DKL{Q}{G_Y}$ for any distribution $Q$ that is absolutely continuous with respect to  $G_Y$.  Applying this inequality to \eqref{eq:W2_to_KL_b} with $Q = P_{Y \mid \Theta}$ leads to the stated result. 
\end{proof}

\subsection{Bound for Distributions on the Sphere} 
Combining Lemma~\ref{lem:W2_to_KL} with the bounds on the relative entropy in Section~\ref{sec:KL_bounds} leads to bounds on the expected Wasserstein distance in terms of $\alpha_r(X)$ and $\beta_r(X)$. The next result leverages Theorem~\ref{thm:DKL_sphere} to give a bound for the setting where $X$ has constant magnitude. 

\begin{theorem}\label{thm:W2_sphere}Consider Assumption 1.  If $\frac{1}{n} \|X\|^2 = \gamma$ almost surely then
\begin{align*}
\MoveEqLeft \ex{ W_2^2(P_{Z \mid \Theta}, G_Z) } \\
& \le  C \,  \gamma\, \,   k^\frac{3}{4} \left( \frac{\frac{1}{n} \|\ex{X}\|^2}{\gamma } \right)^\frac{1}{2}   +  C \, \gamma\,   k\,  \left(  \frac{ \beta_2(X)}{\gamma} \right)^\frac{4}{k +4} ,
\end{align*}
where $C$ is a universal constant. In particular, the inequality holds with $C = 10$.
\end{theorem}
\begin{proof}
By the convexity of Wasserstein distance, we obtain the simple upper bound
\begin{align}
\ex{W_2^2(P_{Z|\Theta}, G_Z)} & \le \ex{W_2^2(P_{Z|\Theta, X}, G_Z)}  \notag \\
& =  \ex{ \| \Theta X\|^2} + k \gamma   \notag \\
& = 2k \gamma. \label{eq:W_sphere_c}
\end{align} 
Alternatively, for $t > 0$, combining Lemma~\ref{lem:W2_to_KL} and Theorem~\ref{thm:DKL_sphere} yields
\begin{align}
 \MoveEqLeft \ex{ W_2^2(P_{Z \mid \Theta}, G_Z) }    \le 4 kt  + 4 (t + \gamma)k^\frac{3}{4} \left( \frac{\frac{1}{n} \|\ex{X}\|^2}{\gamma} \right)^\frac{1}{2} \notag  \\
& \quad   +  4 (t + \gamma)  k^\frac{1}{4} \left(1 + \frac{ 2 \gamma }{t} \right)^\frac{k}{4} \frac{ \beta_2(X)}{\gamma}.\label{eq:W_sphere_d}
\end{align}
Combining the above inequalities leads to
\begin{align*}
\ex{ W_2^2(P_{Z \mid \Theta}, G_Z) }   &  \le 4 kt  + 6  \gamma  k^\frac{3}{4} \left( \frac{\frac{1}{n} \|\ex{X}\|^2}{\gamma}\right)^\frac{1}{2}  \\
& \quad  +  6   \gamma     k^\frac{1}{4} \left(\frac{3}{2} \gamma \right)^\frac{k}{4} \frac{ \beta_2(X)}{\gamma} t^{-\frac{k}{4} }. 
\end{align*}
This inequality holds for $t \ge \gamma/2$ because of \eqref{eq:W_sphere_c} and for $0 < t \le \gamma/2$ because of \eqref{eq:W_sphere_d}. Note that only the first and third terms on the right-hand side depend on $t$. Evaluating this expression with 
\begin{align*}
t^* &  =  \frac{3\gamma }{2}    \left( \frac{ k^\frac{1}{4}   }{4 }  \, \frac{ \beta_2(X)}{\gamma}  \right)^\frac{4}{k  +4 } ,
\end{align*}
leads to
\begin{align*}
4 kt^* +  6   \gamma     k^\frac{1}{4} \left(\frac{3}{2} \gamma \right)^\frac{k}{4} \frac{ \beta_2(X)}{\gamma} (t^*)^{-\frac{k}{4} }
&  = c_k \, k \, \gamma \, \left( \frac{\beta_2(X)}{\gamma} \right)^\frac{4}{k+4} ,
\end{align*}
where
\begin{align*}
c_k & = 6 \left( 1  + \frac{4}{k}    \right)    \left(\frac{ k^\frac{1}{4}   }{ 4 }   \right)^\frac{4}{k  +4 } .
\end{align*}
Finally, it is easy to check that $c_k < 10$ for all $k \ge 1$. 
\end{proof}

\subsection{Proof of Theorem~\ref{thm:W2_general}}

To obtain bounds for general distributions on $X$, one possible approach is combine Lemma~\ref{lem:W2_to_KL}  with Theorem~\ref{thm:DKL_general},  following the same steps used in the proof of Theorem~\ref{thm:W2_sphere}. However, one issue that arises in this approach is that the minimization with respect to $t$ depends on both $\alpha_1(X)$ and $\beta_2(X)$. To bypass this issue, we use a conditioning argument that allows us to apply Theorem~\ref{thm:W2_sphere} to a  projection of $X$ onto the Euclidean sphere.

Given any random vector $X$ that is not deterministically zero and a measurable subset $\set$ of $\reals^n$ that does not include the origin, the vector $X_\set$ is defined according to 
\begin{align}
X_\set = \begin{dcases}
\frac{  \sqrt{  \ex{ \|X\|^2}}   }{\|X\|} X , & X \in \set\\
0, & X \notin \set
\end{dcases} \label{eq:Xprime}.
\end{align}
The next result bounds the expected Wasserstein distance in terms of the conditional distribution of $X_\set$ given $X \in \set$.  

\begin{lemma}\label{lem:truncation_u}
For every measurable set $\set \subseteq \reals^n \backslash \{0\}$ with $P_X(\set) > 0$, the Wasserstein distance satisfies
\begin{align*}
\ex{ W_2^2(P_{Z \mid \Theta}, G_Z) }
& \le  2   k\, \alpha_1(X)  + 2 k \gamma P_X(\set^c) \\
& \quad  + 2  \ex{ W_2^2( P_{Z_\set \mid \Theta, X \in \set}, G_{Z}) } P_X(\set), 
\end{align*}
where $Z_\set = \Theta X_\set $ and $X_\set$ is given by \eqref{eq:Xprime}.
\end{lemma}
\begin{proof}

We begin by focusing on inequalities that hold pointwise with respect to the matrix $\Theta $. We define the distributions
\begin{align*}
P_1  = P_{Z \mid \Theta , X \in \set^c},  \quad P_2  = P_{Z \mid \Theta , X \in \set}, \quad P_3  = P_{Z_\set \mid \Theta , X \in \set},
\end{align*} 
and note that $P_{Z|\Theta} = (1-\lambda) P_1 + \lambda P_2$ where $\lambda = P_X(\set)$. Using this notation, we can now write
\begin{align}
 W^2(P_{Z \mid \Theta} , G_Z)  &  \le (1-\lambda)  W^2(P_1 , G_Z)    +  \lambda W^2(P_2 , G_Z)  \notag \\
 & \le  (1-\lambda)W^2(P_1 , G_Z)    +  2\lambda  W^2(P_2 , P_3) \notag \\
 & \quad  + 2 \lambda W^2(P_3, G_Z), \label{eq:truncation_u_b}
 \end{align}
where the first step follows from the convexity of $W_2^2(P,Q)$ and the second step follows from the triangle inequality for $W_2(P,Q)$ combined with the fact that  $(a+b)^2 \le 2a^2 + 2b^2$. 

To upper bound the first term in \eqref{eq:truncation_u_b}, let $Z^* \sim G_Y$ be independent of $Z$.  From the definition of the Wasserstein distance, we can write
\begin{align*}
W^2(P_1 , G_Z)   &   \le \ex{ \| Z - Z^* \|^2 \;  \middle | \;   \Theta , X \in \set^c} \\
& =  \ex{ \|\Theta X \|^2\,  \middle | \,  \Theta, X \in \set^c}  + k \gamma,
\end{align*}
where the second step follows because $Z^*$ has mean zero and $\ex{ \|Z^*\|^2} = k \gamma$.

To upper bound the second term in \eqref{eq:truncation_u_b}  let  $(Z,Z_\set)$ be defined according to $Z = \Theta X$ and $Z_\set = \Theta X_\set$ where the relationship between $X$ and $X_\set$ is given by \eqref{eq:Xprime}. Conditioned on $\Theta$ and the event $X \in \set$, we have $Z \sim P_2$ and $Z_\set \sim P_3$, and thus
\begin{align*}
 W^2(P_2 , P_3)  
& \le \ex{  \| Z -  Z_\set     \|^2 \; \middle \vert \;  \Theta , X \in \set }\\
& = \ex{  \| \Theta X_\set \|^2 \left( \sqrt{S}  -   \sqrt{\gamma }  \right)^2 \;  \middle \vert \;  \Theta, X \in \set }\\
& \le \ex{  \left \| \Theta X_\set    \right\|^2  | S- \gamma | \,  \middle  | \, \Theta,  X \in \set },
\end{align*}
where $S = \frac{1}{n} \|X\|^2$ and the second inequality follows from noting that
\[
(\sqrt{x} -\sqrt{y})^2 = \frac{(x- y)  (\sqrt {x }- \sqrt{y})}{ \sqrt{x}  +\sqrt{y}} \le |x-y|,
\]
for all $x \ge 0 $ and $y > 0$. 

Collecting the terms back together and taking  the expectation with respect to  $\Theta$ leads to 
\begin{align*}
 \ex{  W^2(P_{Z \mid \Theta} , G_Z) } 
& \le \ex{ \|\Theta X\|^2   \one_{ \set^c}(X)  } + k  \gamma P_X(\set^c) \\
&  +\ex{  \left \| \Theta X_\set     \right\|^2  | S- \gamma | \one_\set(X)   }\\
& \quad +  2 \ex{ W^2(P_{Z_\set | \Theta, X \in \set} , G_Z) } P_X(\set).
 \end{align*}
We can further simplify this bound using the fact that $\|\Theta X\|^2$ can be decomposed as $U  S$ where $U$ is a chi-squared random variable with $k$ degrees of freedom that is independent of $S$. Consequently, 
\begin{align*}
\MoveEqLeft  \ex{ \|\Theta X\|^2   \one_{ \set^c}(X)  }  + \ex{  \left \| \Theta X_\set    \right\|^2  | S- \gamma | \one_\set(X)   }  \\
& = k \ex{ S  \one_{ \set^c}(X)  } + 2 k \ex{ | S- \gamma |  \one_{\set}(X)  }\\
& = k  \ex{ (S- \gamma)   \one_{ \set^c}(X)  } +   k \gamma P_X(\set^c)\\
& \quad + 2 k \ex{ | S- \gamma |  \one_{\set}(X)  }\\
& \le  2 \ex{ |S - \gamma| } + \gamma P_X(\set^c).
\end{align*} 
This completes the proof of Lemma~\ref{lem:truncation_u}.
\end{proof}

We are now ready to prove Theorem~\ref{thm:W2_general}. Let $\set$ be defined according to
\begin{align*}
\set = \left\{ x \in \reals^n \, : \,  \left| \frac{1}{n} \|x\|^2  -\gamma \right|  \le \frac{1}{2}  \gamma \right\}.
\end{align*}
By Markov's inequality,  we see that $P_X(\set^c) \le 2 \alpha_1(X)/\gamma$. 

Next, let $X'$ denote a vector that  is drawn according to the conditional distribution of $X_\set$ given $X \in \set$. By construction, the magnitude satisfies $\|X'\| = \sqrt{n \gamma}$ almost surely and thus, by Theorem~\ref{thm:W2_sphere}, 
\begin{align}
\MoveEqLeft \ex{ W_2^2(P_{Z_\set \mid \Theta , X \in \set }, G_Z) } \notag \\
& \le  C \,  \gamma\, \,   k^\frac{3}{4} \frac{\sqrt{\frac{1}{n} \|\ex{X'}\|^2 }}{\sqrt{\gamma} }   +  C \, \gamma\,   k\,  \left(  \frac{ \beta_2(X')}{\gamma} \right)^\frac{4}{k +4} . \label{eq:W2_general_d} 
\end{align}
The magnitude of the expectation of $X'$ obeys
\begin{align*}
\| \ex{ X'}\| \le \beta_1(X'),
\end{align*}
and the function $\beta_r(X')$ can be related to $\beta_r(X)$ by noting that
\begin{align*}
\MoveEqLeft  \beta_r^r(X')  P_X^2(\set) \\
& =   \ex{ \left| \frac{ \gamma  \langle X_1,X_2 \rangle}{ \|X_1\| \, \|X_2\|}  \right|^r  \; \middle \vert \; X_1 \in \set, X_2 \in \set   } P_X^2(X) \\
& \le   \ex{ \left|  \frac{2}{n}   \langle X_1,X_2 \rangle  \right|^r  \; \middle \vert \; X_1 \in \set, X_2 \in \set   } P_X^2(X) \\
& =  \ex{ \left| \frac{2}{n} \langle X_1,X_2 \rangle  \right|^r  \one_\set( X_1) \one_\set(  X_2)  }  \\
& \le  2^r  \beta_r^r(X). 
\end{align*}
Multiplying both sides of \eqref{eq:W2_general_d} by $P_X(\set)$ and then applying these inequalities gives
\begin{align}
\MoveEqLeft \ex{ W_2^2(P_{Z_\set \mid \Theta, X\in \set }, G_Z) } P_X(\set)  \notag \\
& \le  C \,  \gamma\, \,   k^\frac{3}{4} \frac{\sqrt{ 2 \beta_1(X) }}{\sqrt{\gamma} }   +  C \, \gamma\,   k\,  \left(  \frac{2  \beta_2(X)}{\gamma} \right)^\frac{4}{k +4}. \label{eq:W2_general_e} 
\end{align}
Combining this inequality with Lemma~\ref{lem:truncation_u}  and the upper bound on $P_X(\set^c)$ leads to the stated inequality. This completes the proof of Theorem~\ref{thm:W2_general}.

\section{Conclusion}

The main results of this paper bounds on the deviation between the conditional distribution of the projections of a high-dimensional random vector and a Gaussian approximation, where the conditioning is on the projection matrix. The bounds are given in terms of the quadratic Wasserstein distance and relative entropy  and are stated explicitly as a function of the number of projections and certain key properties of the random vector. In comparison with previous work, one of the contributions of this paper is that our results provide a stronger characterization of the approximation error. 



For the settings considered in this paper, most of the results are essentially the same if the projection matrix is drawn uniformly on the Stiefel manifold (i.e.\, the set of all $k \times n$ matrices satisfying $\Theta \Theta^T = I_k)$. An interesting question for future work is the extent to which our techniques can be applied to more general classes of random projections. 

\appendices

\section{Analysis of Moments}\label{sec:prop_M} 

This section provides bounds on the functions $m_{p}(Y,\Theta)$ and $M(Y, \Theta)$ in terms of properties of the distribution on $X$. To simplify the notation we will write $m_p$ and $M$ where the dependence on $(Y,\Theta)$ is implicit. 

\subsection{Distributions on the Sphere}

To understand the behavior of $m_p$ it is useful to first consider the setting where $X$ has constant magnitude, i.e., $\|X\| = \sqrt{n \gamma}$ almost surely. In this case,  $V_g = V_a = t   + \gamma$ almost surely and it follows from Lemma~\ref{lem:M_kp_characterization} that
\begin{align}
m_p & = \ex{ g_{k,p}\left( \frac{R}{t + \gamma}\right )   } (t + \gamma)^\frac{p-k}{2}, \label{eq:Mkp_g_sphere}
\end{align}
where the function $g_{k,p} : (-1,1) \to \reals$ is defined according to 
\begin{align}
g_{k,p}(u) & =  (1 -u)^{-\frac{k}{2} } (1 + u)^\frac{p}{2} - 1. \label{eq:g_kp} 
\end{align}

The next result provides an upper bound on the term inside the expectation.

\begin{lemma}\label{lem:g_kp_inq}
For all $(r,s,t)$ with  $t, s > 0$ and $|r| \le \gamma$,
\begin{align*}
g_{k,p}\left(\frac{r}{t + \gamma} \right ) &  \le \frac{(k+p)}{2}   \frac{r}{t + \gamma} \\
& \quad  +  \left( t + 2 \gamma \right)^\frac{p}{2} (t + \gamma)^{\frac{k - p}{2} } \frac{r^2}{\gamma^2}.
\end{align*}
\end{lemma}
\begin{proof}
We begin with the decomposition
\begin{align*}
g_{k,p}(u) =g'_{k,p}(0) \, u  +   h_{k,p}(u)  \, u^2,
\end{align*} 
where $g'_{k,p}(0) = (k+p)/2$ and $h_{k,p}(u)= ( g_{k,p}(u)  - u\,  g'_{k,p}(0) ) / u^2$. With a bit of work, it can be verified that $h_{k,p}(u)$ is non-negative and non-decreasing and thus, for all $-1 < u \le z < 1$, 
\begin{align*}
g_{k,p}(u) \le  \frac{(k+p)}{2} \, u  + u^2\,  h_{k,p}(z) .
\end{align*} 
Next, we note that for $z > 0$, 
\begin{align*}
h_{k,p}(u)  & = \frac{1}{z^{2}}  \left( g_{k,p}\left(  z  \right)  - \tfrac{k+p}{2}  z \right) \le \frac{1}{z^2} (1 - z)^{-\frac{k}{2} }(1 + z)^\frac{p}{2}. 
\end{align*}
Combining the above inequalities with $u = r/(t + \gamma)$ and $z = \gamma / (t + \gamma)$ leads to the stated result. 
\end{proof}

\begin{lemma} If $p \ge 0$ and $\|X\| = \sqrt{ n \gamma}$ almost surely then
\begin{align}
m_{p} & \le     \frac{(k+p)}{2}     \left(t + \gamma \right )^\frac{p-k}{2} \frac{  \frac{1}{n} \| \ex{X}\|^2}{ \gamma}  \notag \\
& \quad +      \left(1 +  \frac{2 \gamma}{t}\right )^\frac{p}{2} \frac{ \beta^2_2(X)  }{\gamma^2}. \label{eq:Mkp_g_sphere_c}
\end{align}
\end{lemma}
\begin{proof}
This result follows  from \eqref{eq:Mkp_g_sphere} and Lemma~\ref{lem:g_kp_inq}.
\end{proof}

\subsection{General Distributions}

For general distributions on $X$, we can bound $m_{p}$ using an expression that is similar to \eqref{eq:Mkp_g_sphere}. For $p \ge 0$, the term inside the expectation in Lemma~\ref{lem:M_kp_characterization} satisfies
\begin{align}
\MoveEqLeft \left( \frac{1}{V_a  - R} \right)^\frac{k}{2} \bigg( \frac{ V_g^2 - R^2}{ V_a- R} \bigg)^\frac{p}{2} - \left( \frac{1}{V_a } \right)^\frac{k}{2} \bigg( \frac{ V_g^2}{ V_a} \bigg)^\frac{p}{2} \notag \\
& = \left( \left(1 -  \frac{R}{V_a}  \right)^{ - \frac{k+p}{2}} \left( 1 - \frac{ R^2}{ V_g^2} \right)^\frac{p}{2}  - 1 \right)  V_a^{- \frac{k+p}{2}} V_g^p \notag\\
& \le \left( \left(1 -  \frac{R}{V_a}  \right)^{ - \frac{k+p}{2}} \left( 1 - \frac{ R^2}{ V_a^2} \right)^\frac{p}{2}  - 1 \right)  V_a^{- \frac{k+p}{2}} V_g^p \notag\\
& =  \left( \left(1 -  \frac{R}{V_a}  \right)^{ - \frac{k}{2}} \left( 1 + \frac{ R}{ V_a} \right)^\frac{p}{2}  - 1 \right)  V_a^{ \frac{p-k}{2}}.  \notag \\
& = g_{k,p} \left( \frac{R}{V_a} \right) V_a^{- \frac{k+p}{2}} V_g^p, \label{eq:g_kp_RV} 
\end{align}
where the inequality follows from that fact that $V_g \le V_a$ and the fact that $(1 - R^2 / V_g^2)^\frac{p}{2}$ is non-decreasing in $V_g$. 

Combining this inequality with Lemma~\ref{lem:M_kp_characterization} and  Lemma~\ref{lem:g_kp_inq} leads to our next result.

\begin{lemma} If $p \ge 0$ and $\ex{ \|X\|^\frac{p}{2}} \le \infty$, then
\begin{align}
m_{p} & \le  \frac{(k+p)}{2} \ex{ V_a^{-\frac{k+p}{2}} V_g^p \frac{R}{V_a} }  + \ex{ (t + S)^\frac{p}{2}  \frac{R^2}{S^2} }, \label{eq:Mkp_quad_bound}
\end{align}
where $S = \frac{1}{2n} \|X_1\|^2 + \frac{1}{2n} \|X_2\|^2$. 
\end{lemma}

Note that if $X$ is equal in distribution to $-X$, then the first term in \eqref{eq:Mkp_quad_bound} is equal to zero. 

\subsection{Proof of Lemma~\ref{lem:M_kq_bound}} \label{sec:lem:M_kq_bound_proof}
Combining \eqref{eq:Mkp_quad_bound} with the fact that $\Pmin \le S \le \Pmax$ yields 
\begin{align*}
\MoveEqLeft m_{k-1} \le (1 + \Pmin)^{-\frac{1}{2} } \\
& \quad \times 
\left[  \frac{(2 k -1)}{2}  \frac{ \beta_1(X)}{\Pmin}  + \frac{\beta_2^2(X)}{ \Pmin^2} \left(1 + \frac{ 2 \Pmax}{t} \right)^\frac{k}{2}\right],
\end{align*}
and also 
\begin{align*}
\MoveEqLeft m_{k+1} \le (1 + 2 \Pmax)^{\frac{1}{2} } \\
& \quad \times 
\left[  \frac{(2 k +1)}{2}  \frac{ \beta_1(X)}{\Pmin}  + \frac{\beta_2^2(X)}{ \Pmin^2} \left(1 + \frac{ 2 \Pmax}{t} \right)^\frac{k}{2}\right].
\end{align*}
Taking the geometric mean of these quantities leads to the stated result.  This concludes the proof of Lemma~\ref{lem:M_kq_bound}.

\subsection{Proof of Lemma~\ref{lem:M_k1_bound}} \label{sec:lem:M_k1_bound_proof}
Starting with \eqref{eq:g_kp_RV} and letting $S = \frac{1}{n}\|X_1\|^2 + \frac{1}{n} \|X\|^2$ allows us to write
\begin{align}
m_{p} & \le \ex{ g_{k,p} \left( \frac{R}{V_a} \right) V_a^{- \frac{k+p}{2}} V_g^p} \notag \\
& \le  \ex{ g_{k,p} \left( \frac{|R|}{t +S} \right) (t + S)^{\frac{p-k}{2}}}. \label{eq:Mkp_g_inq}
\end{align}
The next step in the proof is to obtain bounds on term $g_{k,p}( \frac{ r}{ t + s})(t+s)^\frac{p-k}{2}$ that hold for all $s > 0$ and $0 < r \le s$. For $k=1$ and $p=0$, we can write
\begin{align*}
g_{1,0} \left( \frac{r}{ t + s} \right)  \frac{1}{ \sqrt{t + s}} 
& = \left( \sqrt{ 1 + \frac{r}{ t + s -r} } - 1 \right) \frac{1}{ \sqrt{t + s}}\\
& \le  \frac{r}{2( t + s -r)}  \frac{1}{ \sqrt{t + s}}\\
& \le  \frac{r}{2 t \sqrt{t}},
\end{align*}
where the first inequality follows from the fact that $\sqrt{1 + x} -1\le x/2$ and the second inequality follows because the expression is non-increasing in $s$ over $[r,\infty)$. Combining this inequality with \eqref{eq:Mkp_g_inq} gives 
\begin{align*}
m_0 & \le \frac{ \ex{ |R|}}{2 t \sqrt{t} } =  \frac{\beta_1(X) }{2 t \sqrt{t} }. 
\end{align*}
Alternatively, for $k=1$ and $p=2$, we have 
\begin{align*}
g_{1,2} \left( \frac{r}{ t + s} \right)   \sqrt{t + s} 
& = \frac{ t + s + r}{ \sqrt{ t + s - r} }- \sqrt{ t + s}.
\end{align*}
By differentiation, it can be verified that this term is non-increasing in $s$ over the interval $[r,\infty)$, and thus
\begin{align*}
g_{1,2} \left( \frac{r}{ t + s} \right)   \sqrt{t + s} 
& \le  \frac{ t + 2 r}{ \sqrt{ t} }- \sqrt{ t + r} \le \frac{ 2 r}{ \sqrt{t }}.
\end{align*}
Combining this inequality with \eqref{eq:Mkp_g_inq} gives
\begin{align*}
m_2 & \le \frac{ 2 \ex{ |R|}}{\sqrt{t} } = \frac{ 2 \beta_1(X) }{\sqrt{t} }.
\end{align*}
This conclude the proof of Lemma~\ref{lem:M_k1_bound}

\section{Auxiliary Results}\label{sec:auxiliary_results}

\begin{lemma}\label{lem:log_dev_inq}
If $X$ is a non-negative random variable with mean $\mu > 0$, then the following inequality holds for every measurable subset $\set \subseteq \reals$: 
\begin{align*}
\ex{ \log\left( \frac{ 1 + \mu}{1 + X} \right)\one_{\set}(X)}
&  \le \frac{  \log(1 + \mu)}{\mu}   \ex{ |\mu - X|} .
\end{align*} 
\end{lemma}
\begin{proof}
For all $x \ge 0$ and  $0 < \lambda < \mu$, we can write
\begin{align*}
\MoveEqLeft \log\left(\frac{1 + \mu}{1+x} \right)\one_{\set}(x)\\
& \le  \log\left(\frac{1 + \mu}{1+x} \right)\one_{[0,\mu]}(x)\\
&  = \log\left(\frac{1 + \mu}{1+x} \right)\one_{[0, \lambda) }(x)  + \log\left(1 + \frac{ \mu - x}{1+x} \right)\one_{[\lambda,\mu]}(x) \\
&  \le \log\left(1 + \mu \right)\one_{[0, \lambda) }(x)  + \log\left(1 + \frac{\mu - x }{1+\lambda} \right)\one_{[\lambda,\mu]}(x) \\
&  \le \log\left(1 + \mu \right)\left(  \frac{ \mu - x}{ \mu -\lambda} \right)  \one_{[0, \lambda) }(x)  + \left(\frac{ \mu - x}{1+\lambda} \right)\one_{[\lambda,\mu]}(x). 
\end{align*}
Evaluating this inequality with $\lambda= (\mu - \log(1+\mu))/ (1+\log(1+\mu))$ 
leads to
\begin{align*}
 \log\left(\frac{1 + \mu}{1+x} \right)\one_{\set}(x) &  \le \frac{1 + \log(1+\mu)}{1+ \mu} (\mu -x)_+\\
 & \le  \frac{ 2 \log(1+\mu)}{\mu}   (\mu - x)_+ ,
\end{align*}
where the second step follows from the inequality $\mu/(1+\mu) \le \log(1 + \mu)$.  Applying this inequality to $X$ yields 
\begin{align*}
\ex{ \log\left( \frac{ 1 + \mu}{1 + X} \right)\one_{[0,\mu]}(X)} \le   \frac{ 2 \log(1+\mu)}{\mu}  \ex{ (\mu - X)_+} .
\end{align*}
Finally, we use the fact that
\begin{align*}
\ex{ |X - \mu|} & =  \ex{ \max( X - \mu, \mu - X)}\\
&  = \ex{ \mu - X + \max( 2(X - \mu), 0)}\\
& = 2 \ex{ (X- \mu)_+} .
\end{align*}
This completes the proof. 
\end{proof}

\bibliographystyle{IEEEtran}
\bibliography{CCLT_arxiv_v5.bbl}

\begin{thebibliography}{10}
\providecommand{\url}[1]{#1}
\csname url@samestyle\endcsname
\providecommand{\newblock}{\relax}
\providecommand{\bibinfo}[2]{#2}
\providecommand{\BIBentrySTDinterwordspacing}{\spaceskip=0pt\relax}
\providecommand{\BIBentryALTinterwordstretchfactor}{4}
\providecommand{\BIBentryALTinterwordspacing}{\spaceskip=\fontdimen2\font plus
\BIBentryALTinterwordstretchfactor\fontdimen3\font minus
  \fontdimen4\font\relax}
\providecommand{\BIBforeignlanguage}[2]{{%
\expandafter\ifx\csname l@#1\endcsname\relax
\typeout{** WARNING: IEEEtran.bst: No hyphenation pattern has been}%
\typeout{** loaded for the language `#1'. Using the pattern for}%
\typeout{** the default language instead.}%
\else
\language=\csname l@#1\endcsname
\fi
#2}}
\providecommand{\BIBdecl}{\relax}
\BIBdecl

\bibitem{sudakov:1978}
V.~N. Sudakov, ``Typical distributions of linear functionals in
  finite-dimensional spaces of high dimension,'' \emph{Soviet Math. Doklady},
  vol.~16, no.~6, pp. 1578--1582, 1978.

\bibitem{diaconis:1984}
P.~Diaconis and D.~Freedman, ``Asymptotics of graphical projection pursuit,''
  \emph{The Annals of Statistics}, vol.~12, no.~3, pp. 793--815, 1984.

\bibitem{hall:1993}
P.~Hall and K.-C. Li, ``On almost linearity of low dimensional projections from
  high dimensional data,'' \emph{The Annals of Statistics}, vol.~21, no.~2, pp.
  867--889, 1993.

\bibitem{weizsacker:1997}
H.~von Weizs\"acker, ``Sudakov's typical marginals, random linear functionals
  and a conditional central limit theorem,'' \emph{Probability Theory and
  Related Fields}, vol. 107, no.~3, pp. 313--324, 1997.

\bibitem{anttila:2003}
M.~Anttila, K.~Ball, and I.~Perissinaki, ``The central limit problem for convex
  bodies,'' \emph{Transactions of the American Mathematical Society}, vol. 355,
  no.~12, pp. 4723--4735, 2003.

\bibitem{bobkov:2003}
S.~G. Bobkov, ``On concentration of distributions of random weighted sums,''
  \emph{The Annals of Probability}, vol.~31, no.~1, pp. 195--215, 2003.

\bibitem{naor:2003}
A.~Naor and D.~Romik, ``Projecting the surface measure of the sphere of
  $\ell_p^n$,'' \emph{Annales de l'Institut Henri Poincar\'e (B) Probability
  and Statistics}, vol.~39, no.~2, pp. 241--246, 2003.

\bibitem{dasgupta:2006}
S.~Dasgupta, D.~Hsu, and N.~Verma, ``A concentration theorem for projections,''
  in \emph{Conference on Uncertainty in Artificial Intelligence}, 2006.

\bibitem{klartag:2007}
B.~Klartag, ``A central limit theorem for convex sets,'' \emph{Inventiones
  mathematicae}, vol. 168, no.~1, pp. 91--131, April 2007.

\bibitem{klartag:2007a}
------, ``Power-law estimates for the central limit theorem for convex sets,''
  \emph{Journal of Functional Analysis}, vol. 245, no.~1, pp. 284--310, 2007.

\bibitem{meckes:2010aa}
E.~Meckes, ``Approximation of projections of random vectors,'' \emph{Journal of
  Theoretical Probability}, vol.~25, no.~2, pp. 333--352, 2010.

\bibitem{meckes:2012}
------, ``Projections of probability distributions: {A} measure-theoretic
  {D}voretzky theorem,'' in \emph{Geometric Aspects of Functional Analysis},
  ser. Lecture Notes in Mathematics.\hskip 1em plus 0.5em minus 0.4em\relax
  Springer, 2012, vol. 2050, pp. 317--326.

\bibitem{dumbgen:2013}
L.~D\"umbgen and P.~D. Conte-Zerial, ``On low-dimensional projections of
  high-dimensional distributions,'' in \emph{From Probability to Statistics and
  Back: High-Dimensional Models and Processes -- A Festschrift in Honor of Jon
  A. Wellner}.\hskip 1em plus 0.5em minus 0.4em\relax Institute of Mathematical
  Statistics Collections, 2013, vol.~9, pp. 91--104.

\bibitem{leeb:2013}
H.~Leeb, ``On the conditional distributions of low-dimensional projections from
  high-dimensional data,'' \emph{The Annals of Statistics}, vol.~41, no.~2, pp.
  464--483, 2013.

\bibitem{talagrand:1996}
M.~Talagrand, ``Transportation cost for gaussian and other product measures,''
  \emph{Geometric and Functional Analysis}, vol.~6, no.~3, pp. 587--600, 1996.

\bibitem{reeves:2016}
G.~Reeves and H.~D. Pfister, ``The replica-symmetric prediction for compressed
  sensing with {G}aussian matrices is exact,'' 2016, {A}vailable at
  \url{https://arxiv.org/abs/1607.02524}.

\bibitem{reeves:2016a}
------, ``The replica-symmetric prediction for compressed sensing with
  {G}aussian matrices is exact,'' in \emph{Proceedings of the IEEE
  International Symposium on Information Theory (ISIT 2016)}, Barcelona, Spain,
  2016.

\bibitem{reeves:2012a}
G.~Reeves and M.~Gastpar, ``Compressed sensing phase transitions: Rigorous
  bounds versus replica predictions,'' in \emph{Proceedings of the 46-th Annual
  Conference on Information Sciences and Systems (CISS 2012)}, Princeton, NJ,
  March 2012.

\bibitem{minka:2001}
T.~P. Minka, ``Expectation propagation for approximate {B}ayesian inference,''
  in \emph{Proceedings of the 17th Conference in Uncertainty in Artificial
  Intelligence}, ser. UAI '01.\hskip 1em plus 0.5em minus 0.4em\relax San
  Francisco, CA, USA: Morgan Kaufmann Publishers Inc., 2001, pp. 362--369.

\bibitem{opper:2005}
M.~Opper and O.~Winther, ``Expectation consistent approximate inference,''
  \emph{Journal of Machine Learning Research}, vol.~6, pp. 2177--2204, 2005.

\bibitem{guo:2006}
D.~Guo and C.-C. Wang, ``Asymptotic mean-square optimality of belief
  propagation for sparse linear systems,'' in \emph{Proceedings of the IEEE
  Information Theory Workshop}, Chengdu, China, October 2006, pp. 194--198.

\bibitem{donoho:2009a}
D.~L. Donoho, A.~Maleki, and A.~Montanari, ``Message-passing algorithms for
  compressed sensing,'' \emph{Proceedings of the National Academy of Sciences},
  vol. 106, no.~45, pp. 18\,914--18\,919, November 2009.

\bibitem{bayati:2011}
M.~Bayati and A.~Montanari, ``The dynamics of message passing on dense graphs,
  with applications to compressed sensing,'' \emph{IEEE Transactions on
  Information Theory}, vol.~57, no.~2, pp. 764--785, February 2011.

\bibitem{rangan:2011}
S.~Rangan, ``Generalized approximate message passign for estimation with random
  linear mixing,'' in \emph{Proceedings of the IEEE International Symposium on
  Information Theory (ISIT 2011)}, St. Petersburg, Russia, 2011, pp.
  2174--2178.

\bibitem{boom:2015}
W.~van~den Boom, G.~Reeves, and D.~B. Dunson, ``Scalable approximations of
  marginal posteriors in variable selection,'' June 2015, {A}vailable at
  \url{http://arxiv.org/abs/1506.06629}.

\bibitem{boom:2015a}
W.~van~den Boom, D.~B. Dunson, and G.~Reeves, ``Quantifying uncertainty in
  variable selection with arbitrary matrices,'' in \emph{IEEE International
  Workshop on Computational Advances in Multi-Sensor Adaptive Processing
  (CAMSAP)}, 2015.

\bibitem{villani:2003}
C.~Villani, \emph{Topics in Optimal Transportation}, ser. Graduate Studies in
  Mathematics.\hskip 1em plus 0.5em minus 0.4em\relax Providence RI: American
  Mathematical Society, 2003, vol.~58.

\bibitem{csiszar:1984}
I.~Csisz\'ar, ``Sanov property, generalized {$I$}-projection and a conditional
  limit theorem,'' \emph{The Annals of Probability}, vol.~12, no.~3, pp. 768
  --793, 1984.

\bibitem{barron:1986aa}
A.~R. Barron, ``Entropy and the central limit theorem,'' \emph{The Annals of
  Probability}, vol.~14, no.~1, pp. 336--432, 1986.

\bibitem{artstein:2004aa}
S.~Artstein, K.~M. Ball, F.~Barthe, and A.~Naor, ``On the rate of convergence
  in the entropic central limit theorem,'' \emph{Probability Theory and Related
  Fields}, vol. 129, no.~3, pp. 381--390, 2004.

\bibitem{madiman:2007aa}
M.~Madiman and A.~Barron, ``Generalized entropy power inequalities and
  monotonicity properties of information,'' \emph{IEEE Transactions on
  Information Theory}, vol.~53, no.~7, pp. 2317--2329, 2007.

\bibitem{bobkov:2013}
S.~G. Bobkov, G.~P. Chistyakov, and F.~G\"otze, ``Rate of convergence and
  {E}dgeworth-type expansion in the entropic central limit theorem,'' \emph{The
  Annals of Probability}, vol.~41, no.~4, pp. 2479--2512, 2013.

\bibitem{bobkov:2013b}
S.~G. Bobkov, ``Entropic approach to {E}. {R}io's central limit theorem for
  ${W}_2$ transport distance,'' \emph{Statistics and Probability Letters},
  vol.~82, pp. 1644--1648, 2013.

\bibitem{bobkov:2014aa}
S.~G. Bobkov, G.~P. Chistyakov, and F.~G\"otze, ``Berry--{E}sseen bounds in the
  entropic central limit theorem,'' \emph{Probability Theory and Related
  Fields}, vol. 159, no.~3, pp. 343--478, 2014.

\bibitem{bresler:1999}
Y.~Bresler, M.~Gastpar, and R.~Venkataramani, ``Image compression on-the-fly by
  universal sampling in {F}ourier imaging systems,'' in \emph{Proceeding of the
  IEEE Information Theory Workshop}, 1999, pp. 48--48.

\bibitem{donoho:2006a}
D.~L. Donoho, ``Compressed sensing,'' \emph{IEEE Transactions on Information
  Theory}, vol.~52, no.~4, pp. 1289 -- 1306, April 2006.

\bibitem{candes:2006}
E.~J. Cand\`{e}s, J.~Romberg, and T.~Tao, ``Stable signal recovery from
  incomplete and inaccurate measurements,'' \emph{Communications on Pure and
  Applied Mathematics}, vol.~59, pp. 1207--1223, February 2006.

\bibitem{vershynin:2014}
R.~Vershynin, ``Estimation in high dimensions: A geometric perspective,''
  December 2 2014, {A}vailable at
  \url{http://www-personal.umich.edu/~romanv/papers/estimation-tutorial.pdf}.

\bibitem{cover:2006}
T.~M. Cover and J.~A. Thomas, \emph{Elements of Information Theory},
  2nd~ed.\hskip 1em plus 0.5em minus 0.4em\relax Wiley-Interscience, 2006.

\bibitem{topsoe:1967}
F.~Topsoe, ``An information theoretical identity and a problem,'' \emph{Studia
  Scientiarum Mathematicarum Hungarica}, vol.~2, pp. 291--292, 1967.

\bibitem{gibbs:2002}
A.~L. Gibbs and F.~E. Su, ``On choosing and bounding probability metrics,''
  \emph{International Statistical Review}, vol.~70, no.~3, pp. 419--435,
  December 2002.

\bibitem{qi:2010}
F.~Qi, ``Bounds for the ratio of two gamma function,'' \emph{Journal of
  Inequalities and Applications}, 2010.

\bibitem{gupta:1999}
A.~K. Gupta and D.~K. Nagar, \emph{Matrix Variate Distributions}, ser.
  Monographs and Surveys in Pure and Applied Mathematics.\hskip 1em plus 0.5em
  minus 0.4em\relax Chapman and Hall/CRC, 1999.

\end{thebibliography}

\end{document}